\documentclass[a4paper,11pt]{article}

\usepackage{xspace}
\usepackage{algorithm}
\usepackage{algpseudocode}
\usepackage{float}
\usepackage{amsmath,amsthm,bbm,amsfonts,amssymb,url,hyperref}

\usepackage{fullpage}
\providecommand{\keywords}[1]
{
	\small
	\textbf{\textit{Keywords---}} #1
}

\newcommand{\HP}{{\sc Hamiltonian Path}\xspace}
\newcommand{\NP}{\ensuremath{\mathsf{NP}}\xspace}
\newcommand{\NPC}{\ensuremath{\mathsf{NP}}\text{-complete}\xspace}
\newcommand{\NPH}{\ensuremath{\mathsf{NP}}\text{-hard}\xspace}
\newcommand{\PNPH}{para-\ensuremath{\mathsf{NP}\text{-hard}}\xspace}
\newcommand{\el}{\ensuremath{\ell}\xspace}

\newcommand{\FPT}{\ensuremath{\mathsf{FPT}}\xspace}

\newcommand{\APX}{\ensuremath{\mathsf{APX}}\xspace}
\newcommand{\OPT}{{\text{\sc OPT}}\xspace}
\let\oldlambda\lambda
\renewcommand{\lambda}{\ensuremath{\oldlambda}\xspace}
\let\oldalpha\alpha
\renewcommand{\alpha}{\ensuremath{\oldalpha}\xspace}
\let\oldDelta\Delta
\renewcommand{\Delta}{\ensuremath{\oldDelta}\xspace}

\newcommand{\YES}{{\sc yes}\xspace}
\newcommand{\NO}{{\sc no}\xspace}

\usepackage{cleveref}

\newtheorem{observation}{Observation}
\newtheorem{theorem}{Theorem}
\newtheorem{lemma}{Lemma}
\newtheorem{corollary}{Corollary}

\newtheorem{definition}{Definition}

\renewcommand{\AA}{\ensuremath{\mathcal A}\xspace}
\newcommand{\BB}{\ensuremath{\mathcal B}\xspace}
\newcommand{\CC}{\ensuremath{\mathcal C}\xspace}
\newcommand{\DD}{\ensuremath{\mathcal D}\xspace}
\newcommand{\EE}{\ensuremath{\mathcal E}\xspace}
\newcommand{\FF}{\ensuremath{\mathcal F}\xspace}
\newcommand{\GG}{\ensuremath{\mathcal G}\xspace}
\newcommand{\HH}{\ensuremath{\mathcal H}\xspace}
\newcommand{\II}{\ensuremath{\mathcal I}\xspace}

\newcommand{\LL}{\ensuremath{\mathcal L}\xspace}

\newcommand{\OO}{\ensuremath{\mathcal O}\xspace}
\newcommand{\PP}{\ensuremath{\mathcal P}\xspace}
\newcommand{\QQ}{\ensuremath{\mathcal Q}\xspace}

\renewcommand{\SS}{\ensuremath{\mathcal S}\xspace}
\newcommand{\TT}{\ensuremath{\mathcal T}\xspace}
\newcommand{\UU}{\ensuremath{\mathcal U}\xspace}
\newcommand{\VV}{\ensuremath{\mathcal V}\xspace}
\newcommand{\WW}{\ensuremath{\mathcal W}\xspace}

\newcommand{\NB}{\ensuremath{\mathbb N}\xspace}

\newcommand{\pr}{\ensuremath{\prime}}
\newcommand{\prr}{\ensuremath{{\prime\prime}}}
\newcommand{\DEL}{{\text{\sc Del}}\xspace}

\newcommand{\caveat}{\ensuremath{\NP \subseteq \coNP\text{/poly}}\xspace}

\newcommand{\KE}{{\sc Kidney Exchange}\xspace}

\newcommand{\MSTE}{{\sc Max Size $\leq$3-Way Exchange}\xspace}
\newcommand{\TSP}{{\sc 3-Set Packing}\xspace}

\newcommand{\coNP}{\ensuremath{\mathsf{co-NP}}\xspace}
\newcommand{\eps}{\ensuremath{\varepsilon}\xspace}
\newcommand{\tw}{\ensuremath{\tau}\xspace}
\let\mydelta\delta
\renewcommand{\delta}{\ensuremath{\mydelta}\xspace}

\let\mytau\tau
\renewcommand{\tau}{\ensuremath{\mytau}\xspace}

\let\mytheta\theta
\renewcommand{\theta}{\ensuremath{\mytheta}\xspace}

\let\mygamma\gamma
\renewcommand{\gamma}{\ensuremath{\mygamma}\xspace}

\let\myGamma\Gamma
\renewcommand{\Gamma}{\ensuremath{\myGamma}\xspace}
\usepackage{tikz}
\usetikzlibrary{arrows,positioning}

\crefname{theorem}{Theorem}{\bf Theorems}
\crefname{observation}{Observation}{\bf Observations}
\crefname{lemma}{Lemma}{\bf Lemmata}
\crefname{corollary}{Corollary}{\bf Corollaries}
\crefname{proposition}{Proposition}{\bf Propositions}
\crefname{definition}{Definition}{\bf Definitions}
\crefname{claim}{Claim}{\bf Claims}

\title{Parameterized Algorithms for Kidney Exchange}

\author{Palash Dey$^1$ and Arnab Maiti$^2$\\\texttt{palash.dey@cse.iitkgp.ac.in$^1$,arnabm2@uw.edu$^2$}\\ Indian Institute of Technology Kharagpur$^1$, University of Washington$^2$}

\begin{document}

\maketitle

\begin{abstract}
	 In kidney exchange programs, multiple recipient-donor pairs each of whom are otherwise incompatible, exchange their donors to receive compatible kidneys. The \KE problem is typically modelled as a directed graph where every vertex is either an altruistic donor or a pair of recipient and donor; directed edges are added from a donor to its compatible recipients. The computational task is to find if there exists a collection of disjoint cycles and paths starting from altruistic donor vertices of length at most $\el_c$ and $\el_p$, respectively, that covers at least some specific number $t$ of non-altruistic vertices (recipients). We study parameterized algorithms for the kidney exchange problem in this paper. Specifically, we design FPT algorithms parameterized by each of the following parameters: (1) the number of recipients who receive kidney, (2) treewidth of the input graph + $\max\{\el_p,\el_c\}$, and (3) the number of vertex types in the input graph when $\el_p\leq \el_c$. We also present interesting algorithmic and hardness results on the kernelization complexity of the problem. Finally, we present an approximation algorithm for an important special case of the \KE problem.
\end{abstract}


\section{Introduction}

Patients having acute renal failures are typically treated either with dialysis or with kidney transplantation. \textcolor{black}{However, the quality of life on dialysis is comparitively lower and also the average life span of the recipients on dialysis is around 10 years~\cite{usrds}.} For this reason, most recipients prefer a kidney transplantation over periodic dialysis. However, the gap between the demand and supply of kidneys, which can be obtained either from a deceased person or from a living donor, is so large that the average waiting time varies from $2$ to $5$ years at most centers~\cite{unos,odts}. Moreover, even if a recipient is able to find a donor, there could be many medical reasons (like blood group or tissue mismatch) due to which the donor could not donate his/her kidney to the recipient.

The {\em Kidney Paired Donation (KPD)}, a.k.a {\em Kidney Exchange} program, allows donors to donate their kidneys to compatible other recipients with the understanding that their recipients will also receive medically compatible kidneys thereby forming some kind of barter market~\cite{ke04,ke07,AbrahamBS07}. Since its inception in~\cite{rapaport1986case}, an increasing amount of people register in the kidney exchange program since, this way, recipients not only have a better opportunity to receive compatible kidneys, but also can get medically better matched kidneys which last longer~\cite{segev2005kidney}. The central problem in any kidney exchange program, also known as the {\em clearing problem}, is how to transplant kidneys among various recipients and donors so that a maximum number of recipients receive kidneys.

Typically, kidney exchanges happen in a cyclical manner. The problem is represented as a directed graph: a recipient along with his/her donor is a vertex; we add directed edges from a vertex $u$ to all other vertices whose recipients are compatible with the donor of the vertex $u$. In a (directed) cycle, the recipient of every vertex receives a kidney from the donor of the previous vertex along the cycle. Donors do not have any legal obligation to donate kidneys and thus, technically speaking, a donor can leave the program as soon as his/her corresponding recipient receives a kidney~\cite{ke05,segev2005kidney}. This is not only unfair but also leaves a recipient without any donor. To avoid this problem, all kidney transplantations along a cycle are done simultaneously. As each transplantation involves two surgeries, logistic and human resource constraints allow only a few surgeries to be carried out simultaneously. For this reason, most kidney exchange platforms allow transplantation along cycles of very small length only~\cite{ke05,AbrahamBS07,manlove2015paired,mak2017kidney}.

Sometimes we have {\em altruistic donors} (a.k.a {\em non-directed donors (NDDs)}) who do not have recipients paired with them. This allows us to have kidney transplantations along with chains (directed paths) also starting from some altruistic vertices --- an altruistic donor donates his/her kidney to some compatible recipient whose paired donor donates his/her kidney to the next recipient along the chain and so on. Some platform allows non-simultaneous surgeries along a chain since a broken chain is less harmful than a broken cycle --- it does not leave any recipient without donor~\cite{anderson}. However, broken chains also lead to unfairness and thus platforms usually only allow small chains for transplantation. Thus, in the presence of altruistic donors, the fundamental problem of kidney exchanges becomes the following: find a collection of disjoint cycles and chains starting from altruistic vertices which covers a maximum number of non-altruistic vertices. We refer to \Cref{def:prob} in \Cref{sec:prelim} for formal definition of the \KE problem.

\subsection{Related Work}

To the best of our knowledge, Rapaport was the first person to introduce the idea of a kidney exchange~\cite{rapaport1986case}. More than sixty thousands recipients await kidney transplantation and are listed on the United Network for Organ Sharing (UNOS) recipient registry in 2005~\cite{segev2005kidney}. In September, 2004, the Renal Transplant
Oversight Committee of New England approved the establishment of a clearinghouse for
kidney exchanges~\cite{ke05}. Roth et al.~\cite{roth05} discovered a wide class of
constrained-efficient mechanisms that are strategy-proof when recipient–donor pairs and surgeons have $0–1$ preferences.. For example, a line of research allow only cycles~\cite{constantino2013new,klimentova2014new,sonmez2014altruistically} while others allow kidney exchanges to happen along both cycles and chains~\cite{manlove2015paired,glorie2014kidney,xiao2018exact}. All versions of the \KE problem can be formulated as some suitable version of the graph packing problem. Jai et al.~\cite{jia2017efficient} discussed interesting relationship between barter market and set packing.

Krivelevich et al.~\cite{krivelevich2007approximation}, Abraham et al.~\cite{AbrahamBS07}, and Belmonte et al.~\cite{DBLP:journals/algorithmica/BelmonteHKKKKLO22} showed that the basic \KE problem along with its various incarnations are \NPH even when the maximum cycle length allowed is three or the maximum path length allowed is four. Hence, the \KE problem is \PNPH parameterized by the maximum cycle length plus the maximum path length. Krivelevich et al.~\cite{krivelevich2007approximation} and Jia et al.~\cite{jia2017efficient} developed approximation algorithms for the \KE problem by exploiting interesting connection with the set packing problem.

Practical heuristics and integer linear programming based algorithms haven been extensively explored for the \KE problem. For example, Manlove and O'Malley proposed an integer programming-based formulation in order to model the criteria that constitute the definition of optimality as per the laws of United Kingdom~\cite{manlove2015paired}. This formulation was used to build a software which was used by the National Health Service Blood and Transplant to find optimal sets of kidney exchanges for their National Living Donor Kidney Sharing Schemes. Dickerson et al.~\cite{dickerson2016position} argued that although kidney donation along an arbitrarily long chain (path) is desirable in theory, there are many practical concerns because of which one would like to use only bounded length chains for kidney donation. Glorie, Klundert, and Wagelmans~\cite{glorie2014kidney} presented a generic iterative branch-and-price algorithm that can deal effectively with multiple criteria involving cycles and chains. Li et al.~\cite{li2014egalitarian} proposed a practically efficient polynomial time algorithms for finding the Lorenz-dominant fractional matching. The Lorenz-dominant fractional matching, which can be implemented as a lottery of integral matchings, is in some sense the fairest allocation and also enjoys the property of being incentive compatible. Biro et al.~\cite{biro2009maximum} designed approximation algorithms when cycles of length at most three are allowed and no chains are allowed. Riascos-Alvarez et al.~\cite{klimentova2014new} developed a decomposition method that is able to consider long cycles and long chains for projected large realistic instances. In particular, their algorithm also allows the prioritization of the solution composition, for example, chains over cycles or vice versa.

Dickerson et al.~\cite{dickerson2016position} introduced the notion of ``vertex type" and showed its usefulness as a graph parameter in real-world instances of a  kidney exchange program. Two vertices are said to have the same vertex type if their neighbourhoods are the same.

The closest predecessor of our work is by Xiao and Wang~\cite{xiao2018exact}. They proposed an exact algorithm with running time $\OO(2^n n^3)$ where $n$ is the number of vertices in the underlying graph. They also presented a fixed parameter tractable algorithm for the \KE problem parameterized by the number of vertex types if we do not have any restriction on the length of cycles and chains. Lin et al.~\cite{lin2019randomized} studied the version of the \KE problem which allows only cycles and developed a randomized parameterized algorithm with respect to the parameter being the number of patients receiving a kidney and the maximum allowed length of any cycle, combined.

\subsection{Contribution}

Designing exact algorithms for the \KE problem has been a research focus in algorithmic game theory. We contribute to this line of research in this paper. Our specific contribution are the following.

	We design \FPT algorithms for the \KE problem parameterized by the number of recipients receiving kidneys~[\Cref{thm:fpt_t}], treewidth of the underlying graph + maximum length of path($\el_p$) + maximum length of cycle allowed ($\el_c$)~[\Cref{thm:fpt_tw-1}], and the number of vertex types when $\el_p\leq \el_c$~[\Cref{thm:fpt_theta}]. We also present a Monadic second-order ($MSO_2$) formula for the \KE problem where the length of the formula is upper bounded by a function of $\el=\max\{\el_c,\el_p\}$~[\Cref{thm:fpt_tw}].

	We show that the \KE problem admits a polynomial kernel with respect to the number of recipients receiving kidneys + maximum degree when $\max\{\el_p,\el_c\}$ is a constant ~[\Cref{thm:kernel}]. We complement this result by showing that the \KE problem does not admit any polynomial kernel parameterized by the number of recipients receiving kidneys+maximum degree+$\max\{\el_p,\el_c\}$ unless $\NP \subseteq \coNP\text{/poly}$~[\Cref{thm:no-kernel}]

Finally, we also design a $(\frac{16}{9}+\epsilon)$-approximation algorithm for the \KE problem if only cycles of length at most $3$ are allowed (and no paths are allowed)~[\Cref{approx}]. We believe that our work substantially improves the current theoretical understanding of the \KE problem.

A short version of this work has been published in IJCAI 2022~\cite{DBLP:conf/ijcai/MaitiD22}.
\section{Preliminaries}\label{sec:prelim}

For an integer $k$, we denote the sets $\{0,1,\ldots,k\}$ and $\{1,2,\ldots,k\}$ by $[k]_0$ and $[k]$, respectively.

A \KE problem is formally represented by a directed graph $\GG=(\VV,\AA)$ which is known as the {\em compatibility graph}. A subset $\BB\subseteq\VV$ of vertices denotes {\em altruistic donors} (also called {\em non-directed donors}); the other set $\VV\setminus\BB$ of vertices denote a recipient-donor pair who wish to participate in the kidney exchange program. We have a directed edge $(u,v)\in\AA$ if the donor of the vertex $u\in\VV$ has a kidney compatible with the recipient of the vertex $v\in\VV\setminus\BB$. Kidney exchanges happen either (i) along a {\em trading-cycle} $u_1, u_2, \ldots, u_k$ where the recipient of the vertex $u_i\in\VV\setminus\BB$ receives a kidney from the donor of the vertex $u_{i-1}\in\VV\setminus\BB$ for every $2\leq i\leq k$ and the recipient of the vertex $u_1$ receives the kidney from the donor of the vertex $u_k$, or (ii) along a {\em trading-chain} $u_1, u_2, \ldots, u_k$ where $u_1\in\BB, u_i\in\VV\setminus\BB$ for $2\le i\le k$ and the recipient of the vertex $u_j$ receives a kidney from the donor of the vertex $u_{j-1}$ for $2\le j\le k$. Due to operational reasons, all the kidney transplants along a trading-cycle or a trading-chain should be performed simultaneously. This puts an upper bound on the length \el of feasible trading-cycles and trading-chains. We define the length of a path or cycle as the number of edges in it. The kidney exchange clearing problem is to find a collection of feasible trading-cycles and trading-chains which maximizes the number of recipients who receive a kidney. Formally, it is defined as follows.

\begin{definition}[\KE]\label{def:prob}
	Given a directed graph $\GG=(\VV,\AA)$ with no self-loops, an altruistic vertex set $\BB\subseteq\VV$, two integers $\el_p$ and $\el_c$ denoting the maximum length of, respectively, paths and cycles allowed, and a target $t$, compute if there exists a collection \CC of disjoint cycles of length at most $\el_c$ and paths with starting from altruistic vertices only each of length at most $\el_p$ which cover at least $t$ non-altruistic vertices. We denote an arbitrary instance of \KE by $(\GG,\BB,\el_p,\el_c,t)$.
\end{definition}

\subsection{Graph Theoretic Terminologies} In a graph \GG, $\VV[\GG]$ denotes the set of vertices in \GG and $\EE[\GG]$ denotes the set of edges in \GG. Let $\GG[\VV']$ denote the induced subgraph on $\VV'$ where $\VV'\subseteq\VV[\GG]$. Two vertices $u$ and $v$ in a directed graph \GG are called vertices of the same type if they have the same set of in-neighbors and the same set of out-neighbors. \textcolor{black}{If there are no self loops in \GG, vertices of the same type form an independent set.} Treewidth measures how treelike an undirected graph is. We refer to \cite{CyganFKLMPPS15} for an elaborate description of treewidth, tree decomposition, and nice tree decomposition.

\begin{table*}
	\begin{center}
	\begin{tabular}{|cc|}\hline
		\BB & set of altruistic vertices \\
		$t$ & target number of recipients to receive kidneys\\
		$\el_p$ & length of the longest path allowed \\
		$\el_c$ & length of the longest cycle allowed\\
		$\tw$ & treewidth of underlying undirected graph \\
		$\theta$ & number of vertex types\\
		$\Delta$ & maximum degree of underlying undirected graph \\
		$\el$ & $\max\{\el_p,\el_c\}$\\
		\hline
	\end{tabular}
\caption{Notation table.}
\label{table1}
\end{center}
\end{table*}

\begin{definition}[Treewidth]
	Let $G = (V_G,E_G)$ be an undirected graph.  A {\em tree-decomposition} of a graph $G$ is a pair
	$(\mathbb{T} = (V_{\mathbb{T}},E_{\mathbb{T}}),\mathcal{ X}=\{X_{t}\}_{t\in V_{\mathbb T}})$,  where
	${\mathbb T}$ is a tree where every node $t\in V_{\mathbb T}$
	is assigned a subset $X_t\subseteq V_G$, called a bag,  such that the following conditions hold.
	\begin{itemize}
		\item $\bigcup_{t\in V_\mathbb{T}}{X_t}=V_G$,
		\item for every edge $\{x,y\}\in E_G$ there is a $t\in V_\mathbb{T}$ such that  $x,y\in X_{t}$, and
		\item for any $v\in V_G$ the subgraph of $\mathbb{T}$ induced by the set  $\{t\mid v\in X_{t}\}$ is connected.
	\end{itemize}

	The {\em width} of a tree decomposition is $\max_{t\in V_\mathbb{T}} |X_t| -1$. The {\em treewidth} of $G$
	is the  minimum width over all tree decompositions of $G$ and is denoted by $\tw(G)$.

	A tree decomposition  $(\mathbb{T},\mathcal{ X})$ is called a {\em nice edge tree decomposition} if $\mathbb{T}$ is a tree rooted at some node $r$ where $ X_{r}=\emptyset$, each node of $\mathbb{T}$ has at most two children, and each node is of one of the following kinds:
	\begin{itemize}
		\item {\bf Introduce node}: a node $t$ that has only one child $t'$ where $X_{t}\supset X_{t'}$ and  $|X_{t}|=|X_{t'}|+1$.
		\item {\bf Introduce edge node} a node $t$ labeled with an edge between
		$u$ and $v$, with only one child $t'$ such that $\{u,v\}\subseteq X_{t'}=X_{t}$.
		This bag is said to introduce $uv$.
		\item {\bf  Forget vertex node}: a node $t$ that has only one child $t'$  where $X_{t}\subset X_{t'}$ and  $|X_{t}|=|X_{t'}|-1$.
		\item {\bf Join node}:  a node  $t$ with two children $t_{1}$ and $t_{2}$ such that $X_{t}=X_{t_{1}}=X_{t_{2}}$.
		\item {\bf Leaf node}: a node $t$ that is a leaf of $\mathbb T$, and $X_{t}=\emptyset$.
	\end{itemize}
	We additionally require that every edge is introduced exactly once.
	One can  show that  a tree decomposition of width $t$ can be transformed into
	a nice tree decomposition of the same width $t$ and  with
	$\mathcal{O}(t |V_G|)$ nodes, see~e.g.~\cite{BODLAENDER201842,CyganFKLMPPS15}. For a node $t \in \mathbb{T}$, let $\mathbb{T}_t$ be the subtree of $\mathbb{T}$ rooted at $t$, and $V(\mathbb{T}_t)$ denote the vertex set in that subtree. Then $\beta(t)$ is the subgraph of $G$ where the vertex set is  $\bigcup_{t' \in V(\mathbb{T}_t)} X_{t'}$ and the edge set is the union of the set of edges introduced in each $t', t' \in V(\mathbb{T}_t)$. We denote by $V(\beta(t))$ the set of vertices in that subgraph, and by $E(\beta(t))$ the set of edges of the subgraph.

\end{definition}

Since our graph is directed, whenever we mention treewidth of our graph, we refer to the treewidth of the underlying undirected graph; two vertices $u$ and $v$ are neighbors of the underlying undirected graph if and only if either there is an edge from $u$ to $v$ or from $v$ to $u$. Also, refer to the \Cref{table1} for the important notations.

\subsection{Parameterized Complexity}\label{subsec:pc}

A tuple $(x, k)$, where k is the parameter, is an instance of a parameterized problem. \emph{Fixed parameter tractability} (FPT) refers to solvability in time  $f(k) \cdot p(|x|)$ for a given instance $(x, k)$, where  $p$ is a polynomial in the input size $|x|$ and $f$ is an arbitrary computable function of $k$ . We use the notation $\OO^*(f(k))$ to denote $O(f(k)poly(|x|))$.

%
%
We say a parameterized problem is \PNPH if it is \NPH even for some constant values of the parameter.

\begin{definition}[Kernelization]
	A kernelization algorithm for a parameterized problem   $\Pi\subseteq \Gamma^{*}\times \mathbb{N}$ is an
	algorithm that, given $(x,k)\in \Gamma^{*}\times \mathbb{N} $, outputs, in time polynomial in $|x|+k$, a pair
	$(x',k')\in \Gamma^{*}\times
	\mathbb{N}$ such that (a) $(x,k)\in \Pi$ if and only if
	$(x',k')\in \Pi$ and (b) $|x'|,k'\leq g(k)$, where $g$ is some
	computable function.  The output instance $x'$ is called the
	kernel, and the function $g$ is referred to as the size of the
	kernel. If $g(k)=k^{O(1)}$, then we say that
	$\Pi$ admits a polynomial kernel.
\end{definition}

It is well documented that the presence of a polynomial kernel for certain parameterized problems implies that the polynomial hierarchy collapses to the third level (or, more accurately, \caveat{}). As a result, polynomial-sized kernels are unlikely to be present in these problems. We use cross-composition to demonstrate kernel lower bounds.

\begin{definition}[Cross Composition]
	Let $L\subseteq \Sigma^*$ be a language and $Q\subseteq \Sigma^* \times \NB$ be a parameterized language. We say that $L$ cross-composes into $Q$ if there exists a polynomial equivalence relation $R$ and an algorithm $\AA$, called the cross-composition, satisfying the following conditions. The algorithm $\AA$ takes as input a sequence of strings $x_1, x_2 , \ldots , x_t \in \Sigma^*$ that are equivalent with respect to $R$, runs in time polynomial in $\sum_{i=1}^t |x_i|$, and outputs one instance $(y, k)\in \Sigma^*\times\NB$ such that:
	\begin{enumerate}
		\item $k\le p(\max_{i=1}^t |x_i| + \log t)$ for some polynomial $p(\cdot)$, and
		\item $(y,k)\in Q$ if and only if there exists at least one index $i$ such that $x_i\in L$.
	\end{enumerate}
\end{definition}

\section{Results}

We present our results in this section.

\subsection{\FPT Algorithms}

We begin with presenting an \FPT algorithm for the \KE problem parameterized by the number of recipients who receive a kidney. Note that, although the number $t$ of recipients who receive a kidney may not be small, this number is always smaller than the total number $n$ of donors. Thus an \FPT algorithm parameterized by $t$ will always be efficient on some instances of the problem compared to the $\OO(2^n n^3)$ time algorithm of Xiao and Wang~\cite{xiao2018exact}. We use the technique of color coding~\cite{alon1994color,alon1995color} to design our algorithm.

\begin{theorem}\label{thm:fpt_t}
	There is an algorithm for the \KE problem which runs in time $\OO^*(2^{\OO(t)})$.
\end{theorem}

\begin{proof}
	Let $(\GG,\BB,\el_p, \el_c,t)$ be an arbitrary instance of the \KE problem. If $\el_p\geq t$, we check whether there exists a path starting from an altruistic vertex of length $t$; if $\el_c\ge t$, then we check whether there exists a cycle of length $\el_1$ for some $t\leq\el_1\leq \el_c$. Note that this can be checked by a deterministic algorithm in time $\OO^*\left(2^{\OO(t)}\right)$ (for path refer to~\cite{CyganFKLMPPS15} and for cycle refer to~\cite{zehavi2016randomized}). If there exists such a cycle or path, then clearly $(\GG,\BB,\el_p, \el_c,t)$ is a \YES instance. So for the rest of the proof, let us update $\ell_p$ and $\el_c$ as  $\min\{\el_p,t-1\}$ and $\min\{\el_c,t-1\}$, respectively. 

	We observe that if $(\GG,\BB,\el_p,\el_c,t)$ is a \YES instance, then there is a collection $\CC=(\DD_1,\ldots,\DD_k)$ of $k\leq t$ disjoint cycles and paths starting from altruistic vertices each of length at most $\el_c$ and $\el_p$, respectively, which covers $t_1$ non-altruistic vertices where $t\leq t_1\leq 2t$. Note that such a collection will exist as $\el_c<t$ and $\el_p<t$.  We observe that the total number of vertices involved in $\DD_1,\ldots,\DD_k$ is at most $3t$ (at most one altruistic vertices in each $\DD_i, i\in[k]$). We color each vertex of \GG uniformly at random from a set $\SS$ of $3t$ colors. We say a coloring of \GG is good if every vertex in the $\CC$ gets a different color. A random coloring is good with probability at least
	\[\frac{(3t)!}{(3t)^{3t}}\geq \frac{1}{e^{3t}}. \]

	Let $C$ be a non-empty set of colors. Let $\GG_C$ denote an induced subgraph of \GG on the set of vertices colored with one of the colors in $C$. We now solve the \KE problem using dynamic programming. We maintain a dynamic programming table \DD which is indexed by a set of colors $C$ and a number $t'$. $\DD(C,t')$ denotes whether there is a collection of disjoint cycles and paths starting from altruistic vertices each of length at most $\el_c$ and $\el_p$, respectively, which covers $t'$ non-altruistic vertices in the graph $\GG_C$ and no two vertices in this collection have the same color. Now we introduce a function $f$ which has a set of colors $C'$ and a number $i$ as its arguments. $f(C',i)$ decides whether there is a valid colourful cycle or path starting from altruistic vertex of length $i$ in $\GG_{C'}$. By valid colourful cycle (resp., path) we mean that no two vertices in the cycle (resp., path) have the same color and the length of the cycle (resp., path)  is at most $\el_c$ (resp., $\el_p$). $f(C',i)$ can be computed in time $\OO^*(2^{\OO(i)})$ time (refer to~\cite{CyganFKLMPPS15}). Now we present a recursive formula to compute each entry of the table.
\[\DD(C,t')=\bigvee\limits_{C',i: C'\subseteq C,C'\neq \emptyset, 1\leq i\leq t'}\left(\DD(C\setminus C',t'-i)\wedge f(C',i)\right)\]

For base cases, let $\DD(C,0)=1$ and $\DD(\emptyset,t')=0$ if $t'>0$. Now argue the correctness of the above recursive equation. In one direction, let $\DD(C,t')=1$. It implies that there is a collection of valid colourful disjoint cycles and paths starting from altruistic vertices which covers $t'$ non-altruistic vertices in the graph $\GG_C$. Now consider once such disjoint cycle or path. Let the number of non-atruistic vertices in it be $i$ and the set of colors of the vertices in it be $C'$. Then, clearly $f(C',i)=1$ and $\DD(C\setminus C',t-i)=1$. In the other direction, let there exist a set $C'\subseteq C$ and a number $i\in [t']$ such that $\DD(C\setminus C',t'-i)\wedge f(C',i)=1$. This implies that there is a collection of valid colourful disjoint cycles and paths starting from altruistic vertices which covers $t'-i$ non-altruistic vertices in the graph $\GG_{C\setminus C'}$ and there is a valid colourful cycle or path starting from altruistic vertex of length $i$ in $\GG_{C'}$. Hence there is is a collection of valid colourful disjoint cycles and paths starting from altruistic vertices which covers $t'$ non-altruistic vertices in the graph $\GG_{C}$. Therefore $\DD(C,t')=1$.

If $(\GG,\BB,\el_p,\el_c,t)$ is a \YES instance and coloring is good, then $\bigvee\limits_{t\leq t'\leq 2t}\DD(\SS,t')=1$. For a \NO instance, $\bigvee\limits_{t\leq t'\leq 2t}\DD(\SS,t')=0$ for any coloring.

	The total number of entries in the table \DD is $2^{O(t)}\cdot t$. Each entry can be computed in time $\OO^*\left(2^{\OO(\el)}\right)$. Hence, with probability at least $e^{-3t}$, our algorithm outputs the correct decision in $\OO^*\left(2^{\OO(t)}\right)$ time. By repeating $\OO(e^{3t})$ times, we find the correct decision with constant probability. The overall running time of our algorithm is $\OO^*\left(2^{\OO(t)}\right)$. The algorithm can be derandomized by using an $(n,3t,3t)$-splitter.
\end{proof}

We next formulate a Monadic second-order ($MSO_2$) formula for the \KE problem where the length of the formula is upper bounded by $\el=\max\{\el_c,\el_p\}$.

\begin{theorem}\label{thm:fpt_tw}
There is a Monadic second-order ($MSO_2$) formula for \KE problem where the length of the formula is upper bounded by a function of $\el=\max\{\el_c,\el_p\}$.
\end{theorem}
\begin{proof}
Let $(\GG=(\VV,\EE),\BB,\el_p,\el_c,t)$ be an arbitrary instance of the \KE problem. We now define the atomic formulae $h(e,u)$ and $t(e,u)$ which will be used in the $MSO_2$ formula. If $u$ is a vertex variable and $e$ is an edge variable, then the formula $h(e,u)$ is true if and only if $u$ is the head of the edge $e$. Similarly, if $u$ is a vertex variable and $e$ is an edge variable, then the formula $t(e,u)$ is true if and only if $u$ is the tail of the edge $e$. Note that in an edge $(u,v)$, $u$ is the tail and $v$ is the head.

Let $u$ be a vertex variable, $e$ be an edge variable and $E$ be an edge set variable. We now define some formulae which will be useful to construct the $MSO_2$ formula. First, we define $out(u,E)$ which is true if and only if the outdegree of $u$ is at most $1$ in the subgraph $\GG'=(\VV,E)$.
\begin{align*}
out(u,E)=&\forall_{e_1\in E}[t(e_1,u)\rightarrow \lnot(\exists_{e_2\in E} ((e_1\neq e_2)\wedge t(e_2,u) ))]
\end{align*}

Next, we define $in(u,E)$ which is true iff the indegree of $u$ is at most 1 in the subgraph $\GG'=(\VV,E)$.
\begin{align*}
in(u,E)=&\forall_{e_1\in E}[h(e_1,u)\rightarrow \lnot(\exists_{e_2\in E}((e_1\neq e_2)\wedge h(e_2,u) ))]
\end{align*}

Next, we define $alt(u,E)$ which is true iff in the subgraph $\GG'=(\VV,E)$, $u$ is an altruistic vertex or $u$ has outdegree equal to 0 or $u$ has indegree at least 1.
\begin{align*}
alt(u,E)=&(\lnot\exists_{e\in E}\text{ } t(e,u))\vee( \exists_{e\in E}\text{ } h(e,u))\vee(\exists_{v\in \BB}\text{ } u=v)
\end{align*}

Next, for $i\geq 2$, we define $dis_i(x_1,\ldots,x_i)$ which is true iff $x_1,\ldots,x_i$ are pairwise distinct.
\begin{align*}
dis_i(x_1,\ldots,x_i)=\bigwedge_{j\in[i]}\bigwedge_{k\in[i]\setminus\{j\}}(x_j\neq x_k)
\end{align*}

Let the subgraph $\GG'=(\VV,E)$ be a collection of isolated vertices, disjoint paths and disjoint cycles. Then, for $i\in\{1,2,\ldots,l_p-1\}$, we define $path_i(e,E)$ which is true iff there is a disjoint path $p$ in the subgraph $\GG'=(\VV,E)$ which contains the edge $e$ and there are $i$ edges between the head of $e$ and the ending vertex of $p$.
\begin{align*}
path_i(e,E)=& \exists_{e_1,e_2,\ldots,e_i\in E}\exists_{x_1,x_2,\ldots,x_{i+1},x^\pr\in \VV}[dis_{i+2}(x_1,\ldots,x_{i+1},x^\pr)\\
&\wedge (t(e,x^\pr)\wedge h(e,x_1)) \wedge (t(e_1,x_1)\wedge h(e_1,x_2))\\
&\wedge \ldots  \wedge (t(e_i,x_i)\wedge h(e_i,x_{i+1})) \wedge  \lnot \exists _{e'\in E} t(e',x_{i+1}) ]
\end{align*}

Similarly, we define $path_0(e,E)$ as follows.
\begin{align*}
path_0(e,E)=\exists_{x_1\in \VV}(h(e,x_{1})\wedge \lnot \exists _{e'\in E} t(e',x_{1}))
\end{align*}

Let the subgraph $\GG'=(\VV,E)$ be a collection of isolated vertices, disjoint paths and disjoint cycles. Then, for $i\in\{1,2,\ldots,l_c-1\}$, we define $cycle_i(e,E)$ which is true iff there is a cycle $c$ in the subgraph $\GG'=(\VV,E)$ which contains the edge $e$ and there are $i+1$ edges in the cycle $c$.
\begin{align*}
cycle_i(e,E)=& \exists_{e_1,e_2,\ldots,e_i\in E}\exists_{x_1,x_2,\ldots,x_{i+1}\in \VV}[dis_{i+1}(x_1,\ldots,x_{i+1})\\
&\wedge (t(e,x_{i+1})\wedge h(e,x_1)) \wedge (t(e_1,x_1)\wedge h(e_1,x_2))\\
&\wedge \ldots  \wedge (t(e_i,x_i)\wedge h(e_i,x_{i+1})]
\end{align*}

We also define $path_{-1}(e,E),cycle_{-1}(e,E),cycle_{0}(e,E)$ as false for all the edge variables $e$ and for all the edge set variables $E$.

Next, we define our main $MSO_2$ formula which is $\varphi(\BB,E)$ as follows
\begin{align*}
&\varphi(\BB,E)=\forall_{v\in \VV} (out(v,E)\wedge in(v,E)\wedge alt(v,E))\wedge \\
&\forall_{e\in E}\left(\bigvee_{i\in \{-1,0,\ldots,\el_p-1\}}path_i(e,E)\vee \bigvee_{i\in \{-1,0,\ldots,\el_c-1\}}cycle_i(e,E)\right)
\end{align*}

We now show that $\varphi(\BB,E)$ is true iff the subgraph $\GG'=(\VV,E)$ excluding the isolated vertices is a collection of disjoint paths which start with an altruistic vertex and disjoint cycles of length at most $\el_p$ and $\el_c$, respectively.

In one direction, let  $\varphi(\BB,E)$ be true. Then $\forall_{v\in \VV} (out(v,E)\wedge in(v,E)\wedge alt(v,E))$ is true. This implies that the indegree and outdegree of every vertex is at most 1 and disjoint paths, if any begin with an altruistic vertex which implies that $\GG'=(\VV,E)$ is a collection of isolated vertices, disjoint paths beginning with an altruistic vertex and disjoint cycles. Consider any disjoint path $P$ with an edge $e$ whose tail is the starting vertex of $P$. Since  $\varphi(\BB,E)$ is true, $\exists i\in \{0,\ldots,\el_p-1\}$ such that $path_i(e,E)$ is true. This implies that the length of $P$ is at most $\el_p$. Similarly consider any disjoint cycle $C$ with an edge $e$. Since  $\varphi(\BB,E)$ is true, $\exists i\in \{1,\ldots,\el_c-1\}$ such that $cycle_i(e,E)$ is true. This implies that the length of $C$ is at most $\el_c$.

In the other direction, let the subgraph $\GG'=(\VV,E)$ excluding the isolated vertices is a collection of disjoint paths which start with an altruistic vertex and disjoint cycles of length at most $\el_p$ and $\el_c$. This implies that the indegree and outdegree of every vertex is at most 1 and disjoint paths begin with an altruistic vertex which implies $\forall_{v\in \VV} (out(v,E)\wedge in(v,E)\wedge alt(v,E))$ is true. Consider any edge $e$ in $\GG'=(\VV,E)$. If $e$ is part of a path $P$, then there are $i$ edges between the head of $e$ and the ending vertex of $P$ where $i\in\{0,\ldots,\el_p-1\}$ as the length of $P$ is at most $\el_p$. Similarly,  if $e$ is part of a path $C$, then there are $i+1$ edges in the cycle $c$ where $i\in\{1,\ldots,\el_c-1\}$ as the length of $C$ is at most $\el_c$. Therefore, $\forall_{e\in E}\left(\bigvee_{i\in \{-1,0,\ldots,\el_p-1\}}path_i(e,E)\vee \bigvee_{i\in \{-1,0,\ldots,\el_c-1\}}cycle_i(e,E)\right)$ is true. This implies that $\varphi(\BB,E)$ is true.

Let $\alpha(E)=|E|$ where $E$ is a free monadic edge set variable. Then, there exists an algorithm that in $g(||\varphi||,\tw)\cdot |\VV|$ time finds the maximum value of $\alpha(E)$ for which $\varphi(\BB,E)$ is true, where $g$ is some computable function \cite{arnborg1991easy}. If the maximum value of $\alpha(E)$ is at least $t$, we have a YES instance otherwise we have a NO instance.  Note that $||\varphi||$ is upper bounded by function of $\el$ only. This concludes the proof of this theorem.
\end{proof}

It follows immediately from \Cref{thm:fpt_tw} and \cite{arnborg1991easy} that the \KE problem is FPT parameterized by $\tw$ (treewidth) and $\ell=\max\{\el_p,\el_c\}$. However, the running time that we get is not practically useful. Next we proceed to design an efficient dynamic programming based algorithm with running time $\OO^*(\ell^{O(\tw)} \tw^{ O(\tw )} )$.

\begin{theorem}\label{thm:fpt_tw-1}
There is an algorithm for the \KE problem which runs in
time $\OO^*(\ell^{O(\tw)} \tw^{ O(\tw )} )$. In particular, \KE admits an \FPT algorithm parameterized by $(\max\{\el_p,\el_c\},\tw)$.
\end{theorem}

\begin{proof}
	Let $(\GG,\BB,\el_p,\el_c,t)$ be an arbitrary instance of the \KE problem. Let $\TT = (T, \{X_t\subseteq\VV[\GG] \}_{t'\in V (T)} )$ be a nice tree decomposition of the input
	$n$-vertex graph $\GG$ that has width at most $\tw$. Let \TT be rooted at some node $r$. For a node $t'$ of \TT , let $V_{t'}$ be the union of all the bags present in the subtree of \TT rooted at $t'$, including $X_{t'}$. Let $\GG_{t'}=(V_{t'},E_{t'}=\{e: e$ is introduced in the subtree rooted at $t'\})$  be a subgraph of $\GG[V_{t'}]$. We solve the \KE problem using dynamic programming. At every bag of the tree decomposition, we maintain a dynamic programming table \DD which is indexed by the tuple
	$(u, (P_i)_{i\in[\mu]_0}, (Q_i)_{i\in[\mu]}, (E_i)_{i\in[\mu]}, g:[\mu]\rightarrow\{p, c\}, \LL:[\mu]\rightarrow[\el]_0, a:[\mu]\rightarrow\{0,1\},s:[\mu]\rightarrow X_u\cup\{\perp\}, e:[\mu]\rightarrow X_u\cup\{\perp\})$
	where $u$ is the node in the tree decomposition, $\mu\in[\tw+1]$, $(P_i)_{i\in[\mu]_0}$ is a partition of $X_u$, $Q_i$ is a permutation of $P_i$, and $E_i$ is a set of edges from $\EE[\GG_u]$ having both their end points in $P_i$ for every $i\in[\mu]$. Here $\perp$ is a special symbol which is not part of $\VV[\GG]$. We define $\DD(u, (P_i)_{i\in[\mu]_0}, (Q_i)_{i\in[\mu]},(E_i)_{i\in[\mu]}, g, \LL,a,s,e)$ to be the maximum number of edges present in any subgraph \HH of $\GG_u$ such that all the following conditions hold.

	\begin{itemize}
		\item \HH excluding the isolated vertices is a collection of disjoint paths and disjoint cycles.
		\item Disjoint paths starting with an altruistic vertex have length at most $\el_p$ whereas paths without an altruistic vertex have length at most $\el$.
		\item Disjoint cycles have length at most $\el_c$.
		\item Disjoint paths which do not have a vertex from $X_u$ must begin with an altruistic vertex.
		\item $\VV[\HH]\cap X_u = X_u\setminus P_0$
		\item For each path $P\in \HH$, if $\exists i,j\in[\mu]$ such that $P_i\cap \VV[P]\neq \emptyset$ and $P_j\cap \VV[P]\neq \emptyset$, then $i=j$.
		\item For each cycle $C\in \HH$, if $\exists i,j\in[\mu]$ such that $P_i\cap \VV[C]\neq \emptyset$ and $P_j\cap \VV[C]\neq \emptyset$, then $i=j$.

		\item For every $i\in[\mu]$, if $g(i)=p$, then either $P_i$ consists of only one vertex which is isolated or $\EE[\HH]\cap \EE[\GG_u[P_i]]\subseteq \EE[P]$ where $P$ is a disjoint path in \HH, $Q_i$ is a topological order of $P_i$ w.r.t. $P$, $E_i=\EE[\HH]\cap \EE[\GG_u[P_i]]$,  $P$ has $\LL(i)$ edges and exactly $a(i)$ altruistic vertices. If $s(i)=\perp$, then starting vertex of $P$ is not part of $P_i$. If $s(i)=v$, then starting vertex of $P$ is $v$ where $v\in P_i$. If $e(i)=\perp$, then ending vertex of $P$ is not part of $P_i$. If $e(i)=v$, then ending vertex of $P$ is $v$ where $v\in P_i$.


		\item For every $i\in[\mu]$, if $g(i)=c$ and $Q_i=x_1>\cdots>x_\nu$, then $\EE[\HH]\cap \EE[\GG_u[P_i]]\subseteq \EE[C]$ where $C$ is a disjoint cycle in \HH, $Q_i$ is a topological order of $P_i$ w.r.t. the path created by removing the edge $(x_\nu,x_1)$ from $C$, $E_i=\EE[\HH]\cap \EE[\GG_u[P_i]]$, $C$ has $\LL(i)$ edges.
	\end{itemize}
	Clearly $\DD(r,\emptyset,\emptyset,\emptyset,\emptyset,\emptyset,\emptyset,\emptyset,\emptyset)\geq t$ if and only if $(\GG,\BB,\el_p,\el_c,t)$ is a YES instance.

	We now explain how we update the table entry $\DD(u, (P_i)_{i\in[\mu]_0}, (Q_i)_{i\in[\mu]},(E_i)_{i\in[\mu]}, g, \LL,a,s,e)$. First, we make this table entry $-\infty$ if any of the following conditions are satisfied.
	\begin{itemize}
		\item There exists an $i\in[\mu]$ such that $g(i)=c$ and $a(i)=1$.
		\item There exists an $i\in[\mu]$ such that $g(i)=p$, $a(i)=0$ and $s(i)=\perp$.
		\item There exists an $i\in[\mu]$ such that $g(i)=c$ and $\LL(i)<\max\{2,|P_i|\}$.
		\item There exists an $i\in[\mu]$ such that $g(i)=p$ and $\LL(i)<|P_i|-1$.
		\item There exists an $i\in[\mu]$ such that $P_i$ contains more than $a(i)$ altruistic vertices.
		\item There exists an $i\in[\mu]$ such that $\LL(i)>\el$.
		\item  There exists an $i\in[\mu]$ such that $\LL(i)>\el_c$ and $g(i)=c$.
		\item  There exists an $i\in[\mu]$ such that $\LL(i)>\el_p$ and $a(i)=1$.
		\item There exists an $i\in[\mu]$ such that $Q_i=x_1>\ldots>x_\lambda$, $g(i)=p$, $E_i \nsubseteq \{(x_i,x_{i+1}):i\in\{1,\ldots,\lambda-1\}\}$.
		\item There exists an $i\in[\mu]$ such that $Q_i=x_1>\ldots>x_\lambda$, $g(i)=c$, $E_i \nsubseteq \{(x_i,x_{i+1}):i\in\{1,\ldots,\lambda-1\}\}\cup\{(x_\lambda,x_1)\}$.
	\end{itemize}
	We do this because none of the DP Table indices satisfying the above conditions can lead to a valid solution. We now use the following formulas to compute table entries in a bottom-up fashion.

	\textbf{ Leaf node:} Since the tree-decomposition $(\TT, \{X_u\}_{u\in\VV[\TT]})$ is nice, for every leaf node $u$ of \TT, we have $X_u=\emptyset$ and thus all the DP table entries at $u$ is set to $0$.

	\textbf{ Introduce vertex node:} Let $u\in\TT$ be an introduce-vertex-node with child $u^\pr\in\TT$ such that $X_u = X_{u^\pr}\cup\{w\}$ for some vertex $w\in\VV[\GG]\setminus X_{u^\pr}$. Since no edge is introduced in an introduce-vertex-node, $w$ is an isolated vertex in $\GG_u$. We set the table entry $\DD(u, (P_i)_{i\in[\mu]_0},$ $(Q_i)_{i\in[\mu]},(E_i)_{i\in[\mu]} g, \LL,a,s,e)$ to $\DD(u^\pr, (P_j\setminus\{w\},(P_i)_{i\in[\mu]_0\setminus \{j\}}), (Q_i)_{i\in[\mu]}, (E_i)_{i\in[\mu]}g, \LL,a,s,e)$ if $w\in P_j$ where $j=0$ or $P_j=\{w\}$, $g(j)=p$, $\LL(j)=0$, $s(j)=\{w\}$, $e(j)=\{w\}$, $a(j)=\mathbbm{1}(w\in\BB)$ where $\mathbbm{1}(\cdot)$ is the indicator function otherwise we set it to $-\infty$.

	\textbf{ Introduce edge node:} Let $u\in\TT$ be an introduce-edge-node introducing an edge $(x,y)\in\EE[\GG]$ and $u^\pr$ the child of $u$. That is, we have $X_u = X_{u^\pr}$. For the table entry $\DD(u, (P_i)_{i\in[\mu]_0}, (Q_i)_{i\in[\mu]},(E_i)_{i\in[\mu]}, g, \LL,a,s,e)$, if any of the following conditions hold
	\begin{itemize}
		\item there does not exist any $i\in[\mu]$ such that $x,y\in P_i$,

		\item there exists an $i\in[\mu]$ such that $x,y\in P_i$ and $(x,y)\notin E_i$,
	\end{itemize}
	then we set the entry to $\DD(u^\pr, (P_i)_{i\in[\mu]_0}, (Q_i)_{i\in[\mu]},(E_i)_{i\in[\mu]}, g, \LL,a,s,e)$. Otherwise, let us assume that there exists a $j\in[\mu]$ such that $x,y\in P_j$ and $(x,y)\in E_j$. Intuitively, we compute the appropriate DP table indices at node $u^\pr$ that we can get when we remove the edge $\{x,y\}$ from the current DP table index at node $u$. Then, we choose the maximum value among the DP table entries corresponding to DP table indices at node $u^\pr$ that we computed and add $1$ to it to get the value for the current table entry. Formally, we update the table entry as follows.
	\begin{itemize}

		\item If $g(j)=p$, then suppose we have $Q_j=x_1>\cdots>x_\lambda (=x)>x_{\lambda+1}(=y)>\cdots>x_\nu$ (note that $y$ should appear immediately after $x$ in $Q_j$, otherwise we would have already set the entry to $-\infty$). Let us define $\PP^\pr=(P_i^\pr)_{i\in[\mu+1]_0}, \QQ^\pr=(Q_i^\pr)_{i\in[\mu+1]}, \EE^\pr=(E_i^\pr)_{i\in[\mu+1]}, g^\pr:[\mu+1]\rightarrow\{p,c\}, a^\pr:[\mu+1]\rightarrow\{0,1\}, s^\pr:[\mu+1]\rightarrow X_u\cup \{\perp\}, e^\pr:[\mu+1]\rightarrow X_u\cup \{\perp\}$, and $\LL^k:[\mu+1]\rightarrow[\el]_0$ for any integer $k\in[\el]_0$ as follows.
		\begin{itemize}
			\item $P_j^\pr=\{x_1,\ldots,x_\lambda\}, P_{\mu+1}^\pr=\{x_{\lambda+1},\ldots, x_\nu\}, P_i^\pr=P_i$ for every $i\in[\mu]_0, i\ne j$

			\item $Q_j^\pr=x_1>\cdots>x_\lambda, Q_{\mu+1}^\pr=x_{\lambda+1}>\cdots> x_\nu, Q_i^\pr=Q_i$ for every $i\in[\mu], i\ne j$

			\item $E_j^\pr=E_j\cap(\{x_1,\ldots,x_\lambda\}\times\{x_1,\ldots,x_\lambda\}), E_{\mu+1}^\pr=E_j\cap(\{x_{\lambda+1},\ldots, x_\nu\}\times\{x_{\lambda+1},\ldots, x_\nu\}), E_i^\pr=E_i$ for every $i\in[\mu], i\ne j$

			\item $g^\pr(j)=g^\pr(\mu+1)=p$ and $g^\pr(i)=g(i)$ for every $i\in[\mu], i\ne j$

			\item $a^\pr(\mu+1)=0$ and $a^\pr(i)=a(i)$ for every $i\in[\mu]$
			\item $s^\pr(j)=s(j)$, $s^\pr(\mu+1)=y$ and $s^\pr(i)=s(i)$ for every $i\in[\mu]_0, i\ne j$
			\item $e^\pr(j)=x$, $e^\pr(\mu+1)=e(j)$ and $e^\pr(i)=e(i)$ for every $i\in[\mu]_0, i\ne j$

			\item $\LL^k(j)=k, \LL^k(\mu+1)=\LL(j)-1-k, \LL^k(i)=\LL(i)$ for every $i\in [\mu], i\ne j$
		\end{itemize}
		We now set the entry to
		\[ 1+\max_{k\in[\el]_0} \DD(u^\pr, \PP^\pr, \QQ^\pr, \EE^\pr, g^\pr, \LL^k, a^\pr, s^\pr, e^\pr) \]

		\item If $g(j)=c$, then let us define $g^\pr:[\mu]\rightarrow\{p,c\}$ as $g^\pr(j)=p$ and $g^\pr(i)=g(i)$ for every $i\in[\mu], i\ne j$. Let us also define $s^\pr:[\mu]\rightarrow X_u\cup\{\perp\}$ as $s^\pr(j)=y$ and $s^\pr(i)=s(i)$ for every $i\in[\mu], i\ne j$. Similarly, let us define $e^\pr:[\mu]\rightarrow X_u\cup\{\perp\}$ as $e^\pr(j)=x$ and $e^\pr(i)=e(i)$ for every $i\in[\mu], i\ne j$. We set the entry to $1+\DD(u^\pr, (P_i)_{i\in[\mu]_0}, (Q_i)_{i\in[\mu]},((E_i)_{i\in[\mu], i\ne j}, E_j\setminus\{(x,y)\}), g^\pr, \LL^\pr,a,s^\pr,e^\pr)$ where $\forall i\neq j,\LL^\pr(i)=\LL(i) ,\LL^\pr(j)=\LL(j)-1$.
	\end{itemize}

	\textbf{ Forget vertex node:} Let $u\in\TT$ be a forget node with a child $u^\pr$ such that $X_u$ = $X_{u^\pr} \setminus \{w\}$ for some $w \in X_{u^\pr}$. In this case, we update the table entry $\DD(u, (P_i)_{i\in[\mu]_0}, (Q_i)_{i\in[\mu]},(E_i)_{i\in[\mu]}, g, \LL,a,s,e)$ as follows. For an index $(u^\pr, (P_i^\prr)_{i\in[\mu]_0}, (Q_i^\prr)_{i\in[\mu]},(E_i^\prr)_{i\in[\mu]}, g^\prr, \LL^\prr,a^\prr,s^\prr,e^\prr)$ of the DP table at node $u^\pr$, we define $\DEL(u^\pr, (P_i^\prr)_{i\in[\mu]_0}, (Q_i^\prr)_{i\in[\mu]},(E_i^\prr)_{i\in[\mu]}, g^\prr, \LL^\prr,a^\prr,s^\prr,e^\prr)=(u, \PP^\pr, \QQ^\pr, \EE^\pr, g^\pr, \LL^\pr,a^\pr,s^\pr,e^\pr)$ as follows.

	\begin{itemize}
		\item if $w\in P_0^\prr$, then $\PP^\pr=(P_i^\pr)_{i\in[\mu]_0}, \QQ^\pr=(Q_i^\pr)_{i\in[\mu]},\EE^\pr=(E_i^\pr)_{i\in[\mu]}, P_0^\pr=P_0^\prr\setminus\{w\}, P_i^\pr=P_i^\prr, Q_i^\pr=Q_i^\prr, E_i^\pr=E_i^\prr$ for every $i\in[\mu]$ $g^\pr=g^\prr, \LL^\pr=\LL^\prr, a^\pr=a^\prr, s^\pr=s^\prr,e^\pr=e^\prr$


	\item if there exists $j\in[\mu]$ such that $w\in P_j^\prr$ and $|P_j^\prr|>1$, then
	\begin{itemize}
		\item $\PP^\pr=(P_i^\pr)_{i\in[\mu]}$ where $P_j^\pr=P_j^\prr\setminus\{w\}$, $P_i^\pr=P_i^\prr$ for every $i\in[\mu]_0, i\neq j$
		\item We define $\QQ^\pr=(Q_i^\pr)_{i\in[\mu]}$. Suppose we have $Q_j^\prr=x_1>\cdots>x_\lambda(=w)>\cdots>x_\nu$. Then, we define $Q^\pr_j=x_1>\cdots>x_{\lambda-1}>x_{\lambda+1}>\cdots>x_\nu$ if $\lambda\notin\{1,\nu\}$; $Q^\pr_j=x_2>\cdots>x_\nu$ if $\lambda=1$; $Q^\pr_j=x_1>\cdots>x_{\nu-1}$ if $\lambda=\nu$; $Q^\pr_i=Q_i$ for every $i\in[\mu], i\ne j$
		\item $\EE^\pr=(E_i^\pr)_{i\in[\mu]}$ where $E_j^\pr=E_j^\prr\cap(P_j^\pr\times P_j^\pr)$, $E_i^\pr=E_i^\prr$ for every $i\in[\mu],i\neq j$
		\item  $g^\pr(i)=g^\prr(i)$ for every $i\in[\mu]$
		\item $\LL^\pr(j)=\LL^\prr(j)-(|E_j^\prr|-|E_j^\pr|)$, $\LL^\pr(i)=\LL^\prr(i)$ for every $i\in[\mu],i\neq j$
		\item $a^\pr(i)=a^\prr(i)$ for every $i\in[\mu]$.
		\item $s^\pr(i)=s^\prr(i)$ for every $i\in[\mu],i\neq j$. If $s^\prr(j)=w$, then $s^\pr(j)=\perp$ otherwise $s^\pr(j)=s^\prr(j)$.
		\item $e^\pr(i)=e^\prr(i)$ for every $i\in[\mu],i\neq j$. If $e^\prr(j)=w$, then $e^\pr(j)=\perp$ otherwise $e^\pr(j)=e^\prr(j)$.
	\end{itemize}
	\item  if there exists $j\in[\mu]$ such that $P_j^\prr=\{w\}$, then
	\begin{itemize}
		\item if $a(j)=1$ and $g^\prr(j)=p$ or $a(j)=0$ and $g^\prr(j)=c$, then $\PP^\pr=(P_i^\prr)_{i\in[\mu]_0\setminus \{j\}}, \QQ^\pr=(Q_i^\prr)_{i\in[\mu]\setminus \{j\}},\EE^\pr=(E_i^\prr)_{i\in[\mu]\setminus \{j\}}$, $g^\pr=g^\prr, \LL^\pr=\LL^\prr, a^\pr=a^\prr, s^\pr=s^\prr,e^\pr=e^\prr$
		\item if $a(j)=0$ and $g^\prr(j)=p$, then $\PP^\pr=\QQ^\pr=\EE^\pr=g^\pr=\LL^\pr=a^\pr=s^\pr=e^\pr=\perp$. This ensures that the $\DEL(u^\pr, (P_i^\prr)_{i\in[\mu]_0}, (Q_i^\prr)_{i\in[\mu]},(E_i^\prr)_{i\in[\mu]}, g^\prr, \LL^\prr,a^\prr,s^\prr,e^\prr)$ doesn't correspond to any DP table index.
	\end{itemize}
\end{itemize}
Intuitively, applying $\DEL(.)$ on a DP table index at node $u^\pr$ leads to a DP table index at node $u$ which we get when we remove $w$ from the DP table index at node $u^\pr$.

We update the table entry $\DD(u, (P_i)_{i\in[\mu]_0}, (Q_i)_{i\in[\mu]},(E_i)_{i\in[\mu]}, g, \LL,a,s,e)$ as
\begin{align*}
	&\max\{\DD(u^\pr, (P_i^\prr)_{i\in[\mu]_0}, (Q_i^\prr)_{i\in[\mu]},(E_i^\prr)_{i\in[\mu]}, g^\prr, \LL^\prr,a^\prr,s^\prr,e^\prr):\\
	&\DEL(u^\pr, (P_i^\prr)_{i\in[\mu]_0}, (Q_i^\prr)_{i\in[\mu]},(E_i^\prr)_{i\in[\mu]}, g^\prr, \LL^\prr,a^\prr,s^\prr,e^\prr)\\
	&=(u, (P_i)_{i\in[\mu]_0}, (Q_i)_{i\in[\mu]},(E_i)_{i\in[\mu]}, g, \LL,a,s,e)\}
\end{align*}

\textbf{ Join node:} For a join node $u$, let $u_1,u_2$ be its two children. Note that, we have $X_u=X_{u_1}=X_{u_2}$. We say that the following pair of tuple $((u_1, (P_i^1)_{i\in[\mu_1]_0}, (Q_{i}^1)_{i\in[\mu_1]},(E_{i}^1)_{i\in[\mu_1]}, g^1, \LL^1,a^1,s^1,e^1), (u_2, (P_i^2)_{i\in[\mu_2]_0}, (Q_{i}^2)_{i\in[\mu_2]},(E_{i}^2)_{i\in[\mu_2]}, g^2,$ $ \LL^2,a^2,s^2,e^2))$ is said to be compatible with $(u, (P_i)_{i\in[\mu]_0},(Q_i)_{i\in[\mu]},$
$(E_i)_{i\in[\mu]}, g, \LL,a,s,e)$ if the following happens
\begin{itemize}
	\item $\forall i\in[\mu]$, $\exists i_1,i_2,\ldots,i_{k_1}\in [\mu_1]$ and $\exists i_1^\pr,i_2^\pr,\ldots,i_{k_2}^\pr\in [\mu_2]$ such that:
	\begin{itemize}
		\item $P_i\setminus (P_{i_1}^1\cup P_{i_2}^1\cup\ldots\cup P_{i_{k_1}}^1)\subseteq P_0^1$ and $P_i\setminus (P_{i_1^\pr}^2\cup P_{i_2^\pr}^2\cup\ldots\cup P_{i_{k_2^\pr}}^2)\subseteq P_0^2$
		\item $P_{i_1}^1\cup P_{i_2}^1\cup\ldots\cup P_{i_{k_1}}^1\cup P_{i_1^\pr}^2\cup P_{i_2^\pr}^2\cup\ldots\cup P_{i_{k_2}^\pr}^2= P_i$.
		\item $E_i=E_{i_1}^1\cup E_{i_2}^1\cup\ldots\cup E_{i_{k_1}}^1\cup E_{i_1^\pr}^2\cup E_{i_2^\pr}^2\cup\ldots\cup E_{i_{k_2^\pr}}^2$
		\item Let $\QQ=\{Q_{i_1}^1,\ldots,Q_{k_1}^1,Q_{i_1^\pr}^2,\ldots,Q_{i_{k_2}^\pr}^2\}$. We now define a directed multigraph $\GG_i=(P_i\cup\{a,b\},\EE_i)$ where $\EE_i$ is a multiset. $\forall Q_{j_2}^{j_1}\in \QQ$ we do the following:
		\begin{itemize}
			\item Let $Q_{j_2}^{j_1}=x_1>\cdots>x_{\nu}$. Then, add the edges in $\{(x_i,x_{i+1}):i\in[\nu-1]\}$ to  $\EE_i$ .
			\item if $g^{j_1}(j_2)=c$, then add $(x_{\nu},x_1)$ to $\EE_i$.
			\item if $g^{j_1}(j_2)=p$, $e^{j_1}(j_2)=\perp$, then add $(x_{v},b)$ to $\EE_i$.
			\item if $g^{j_1}(j_2)=p$, $s^{j_1}(j_2)=\perp$,  then add $(a,x_1)$ to $\EE_i$
		\end{itemize}
		Similarly, we define a  directed multigraph $\GG_i^\pr=(P_i\cup\{a,b\},\EE_i^\pr)$ by replacing $\EE_i$ and $\QQ$ by $\EE_i^\pr$ and $\{Q_i\}$, respectively. Then, $\EE_i=\EE_i^\pr$, $\LL^1(i_1)+\ldots+\LL^1(i_{k_1})+\LL^2(i_1^\pr)+\ldots+\LL^2(i_{k_2}^\pr)=\LL(i)$ and $a(i)=\max\{a^1(i_1),\ldots,a^1(i_{k_1}),a^2(i_1^\pr),\ldots,a^2(i_{k_2}^\pr)\}$.
	\end{itemize}
	\item $P_{0}\subseteq P_0^1$ and $P_{0}\subseteq P_0^2$.

\end{itemize}
Intuitively, a pair of tuples is said to be compatible with the current DP table index at node $u$, if partitions in the pair of tuples can be used to construct the tuple at our current DP table index at node $u$. For instance, we check whether we can use the paths corresponding to the partitions in the pair to construct the paths corresponding to the partition in the current DP table index at node $u$.

Let $T^C$ be the set of all pair of tuples $(T_1,T_2)$ which are compatible as per the above conditions with respect to $(u, (P_i)_{i\in[\mu]_0}, (Q_i)_{i\in[\mu]},(E_i)_{i\in[\mu]}, g, \LL,a,s,e)$. Then,
\begin{align*}
	\DD(u, (P_i)_{i\in[\mu]_0}, (Q_i)_{i\in[\mu]},(E_i)_{i\in[\mu]}, g, \LL,a,s,e)
	=\max_{(T_1,T_2)\in T^C}\DD(T_1)+\DD(T_2)
\end{align*}

Each table entry can be updated in $\el^{\OO(\tw)}\tw^{\OO(\tw)}$ time and the size of each table in any node of the tree decomposition is $\el^{\OO(\tw)}\tw^{\OO(\tw)}$. Hence, the running time of our algorithm is $\OO^*(\el^{\OO(\tw)}\tw^{\OO(\tw)})$.
	\end{proof}

We again note that although the treewidth \tw of the underlying undirected graph may not always be small, it is often significantly smaller than the number $n$ of the donors. Hence, the \FPT algorithm of \Cref{thm:fpt_tw-1} will often outperform the exact $\OO(2^n n^3)$ time algorithm of Xiao and Wang~\cite{xiao2018exact}.


Let $\theta$ denote the number of vertex types in a graph $\GG(\VV,\EE)$. Xiao and Wang~\cite{xiao2018exact} presented an \FPT algorithm parameterized by $\theta$ when $\el_p=\el_c=|\VV|$. We improve the result by presenting an \FPT algorithm parameterized by $\theta$ when $\el_p\leq \el_c$. Towards that, we first present an important lemma on the structure of an optimal solution.

\begin{lemma}\label{lem:theta}
	In every \KE problem instance when $\el_p\leq \el_c$, there exists an optimal solution where the length of every path and cycle in that solution is at most $\theta+3$.
\end{lemma}

\begin{proof}
	Let \CC (which is a collection of paths starting from altruistic vertices and cycles) be an optimal solution of an instance \II of the \KE problem with $\el_p\leq \el_c$. If the length of every path and cycle in \CC is at most $\theta+3$, then there is nothing to prove. So let us assume otherwise. That is, there either exists a path starting from an altruistic vertex of length more than $\theta+3$ or a cycle of length more than $\theta+3$ in \CC (or both). We now construct another solution $\CC_1$ as follows.

	{\bf Case I: } Suppose there exists a path $p=v_1 v_2 \ldots v_k$ in \CC where $v_1$ is an altruistic vertex and $\el_p\ge k>\theta+3$. By pigeonhole principal, there exists a vertex type \gamma such that the vertices $v_a$ and $v_b$ are of type \gamma for some indices $1< a<b< k$. Note that $b\neq a+1$ as vertices of the same type form an independent set. Consider a path $p_1$ defined as $p_1: v_1 \ldots v_a v_{b+1}\ldots v_k$ and cycle $q$ defined as $q: v_{a+1}\ldots v_b v_{a+1}$. Since there exists an edge from $v_b$ to $v_{b+1}$ and the vertices $v_a$ and $v_b$ are of same type (of type \gamma), there exists an edge from $v_a$ to $v_{b+1}$. Similarly, since there exists an edge from $v_a$ to $v_{a+1}$ and the vertices $v_a$ and $v_b$ are of same type (of type \gamma), there exists an edge from $v_b$ to $v_{a+1}$. Hence, the path $p_1$ and cycle $q$ is well-defined. We now consider the solution $\CC_1 = (\CC\setminus\{p\})\cup\{p_1, q\}$. We observe that both \CC and $\CC_1$ cover the same number of non-altruistic vertices. $\CC_1$ is a valid solution as length of $p_1$ and $q$ is at most $\el_p$ and we have $\el_p\leq \el_c$.

	{\bf Case II: } Suppose there exists a cycle $q=v_1 v_2 \ldots v_k v_1$ in \CC where $\el_c\ge k>\theta+3$. By pigeonhole principal, there exists a vertex type \gamma such that the vertices $v_a$ and $v_b$ are of type \gamma for some indices $1< a<b< k$. Note that $b\neq a+1$ as vertices of the same type form an independent set. Consider cycles $q_1: v_a v_{a+1}\ldots v_{b-1} v_a$ and $q_2: v_1 \ldots v_{a-1} v_b v_{b+1}\ldots v_k v_1$. Since there is an edge from $v_{b-1}$ to $v_b$ and the vertices $v_a$ and $v_b$ are of same type, there exists an edge from $v_{b-1}$ to $v_a$. Similarly, since there exists an edge from $v_{a-1}$ to $v_a$ and the vertices $v_a$ and $v_b$ are of same type, there exists an edge from $v_{a-1}$ to $v_b$. Hence, the cycle $q_1$ and $q_2$ are well-defined. We now consider the solution $\CC_1 = (\CC\setminus\{q\})\cup\{q_1, q_2\}$. We observe that both \CC and $\CC_1$ cover the same number of non-altruistic vertices. $\CC_1$ is a valid solution as length of $q_1$ and $q_2$ is at most $\el_c$.

	If $\CC_1$ has any path or cycle of length more than $\theta+3$, then we apply the above procedure to construct another optimal solution $\CC_2$. We repeat the procedure as long as the current optimal solution has any path or cycle of length more than $\theta+3$. The above process terminates after finite number of steps with an optimal solution $\CC_s$ since we are always replacing longer path/cycle with shorter paths and cycles. Hence, in the optimal solution $\CC_s$, the length of every path and cycle is at most $\theta+3$ which proves the result.
\end{proof}

We use the following useful result by Lenstra to design our \FPT algorithm.

\begin{lemma}[Lenstra's Theorem~\cite{lenstra1983integer}]\label{lenstra}
There is an algorithm for computing a feasible as well as an optimal solution of an integer linear program which is fixed parameter tractable parameterized by the number of variables.
\end{lemma}

Using \Cref{lem:theta,lenstra}, we develop an \FPT algorithm for \KE parameterized by the number \theta of types of vertices in \Cref{thm:fpt_theta}.

\begin{theorem}\label{thm:fpt_theta}
	There exists an \FPT algorithm for the \KE problem parameterized by $\theta$ when $\el_p\leq\el_c$.
\end{theorem}

\begin{proof}
	Let $(\GG=(\VV,\EE),\BB,\el_p,\el_c,t)$ be an arbitrary instance of the \KE problem where $\el_p\leq\el_c$. We denote the set of types in \GG by \Gamma and the type of vertex $v$ by $\gamma(v)$.  For a type $\gamma\in\Gamma$, let the set of vertices of type \gamma be $V_\gamma\subseteq \VV$. If there is a $\gamma\in\Gamma$ such that $V_\gamma$ contains both altrusitic and non-altruistic vertices, then we remove the non-altruistic vertices as they have in degree $0$ and can't be part of any feasible solution. For a path/cycle $p=v_1 \ldots v_k$, ``signature of $p$'' is defined as $\gamma(v_1)\ldots\gamma(v_k)$ and we denote it by $\gamma(p)$. For a type $\gamma \in \Gamma$, we denote the number of vertices of type \gamma in \GG by $n(\gamma)$. For a type $\gamma \in \Gamma$ and a path/cycle/signature-sequence $p$, we denote the number of vertices of type \gamma in $p$ by $n_p(\gamma)$. For a path/cycle/signature-sequence $p$, we denote the number of non-altruistic vertices in $p$ by $\lambda(p)$.  Let \AA denote the set of signatures of paths starting from altruistic vertices of length at most $\min\{\theta+4,\el_p\}$ and signatures of cycles of length at most $\min\{\theta+3,\el_c\}$ in \GG. Since there are \theta types, we have $|\AA|=\OO(\theta^\theta)$. We can compute the set \AA in time $\OO^*(\theta^\theta)$ as follows. For each possible signature-sequence $p=\gamma_1\ldots\gamma_k$ of a path, we check if there is an edge from a vertex in $V_{\gamma_i}$ to a vertex in $V_{\gamma_{i+1}}$ for all $i\in[k-1]$ and  $n_p(\gamma)\leq n(\gamma)$ for all $\gamma\in\Gamma$. If the conditions hold true, then we add $p$ to the set \AA. Similarly, for each possible signature-sequence $p=\gamma_1\ldots\gamma_k$ of a cycle, we check if there is an edge from a vertex in $V_{\gamma_i}$ to a vertex in $V_{\gamma_{(i\text{ mod } k)+1}}$ for all $i\in[k]$ and  $n_p(\gamma)\leq n(\gamma)$ for all $\gamma\in\Gamma$. If the conditions hold true, then we add $p$ to the set \AA.

We consider the following integer linear program; its variables are $x(p)$ for every $p\in\AA$.
	\begin{align*}
		\text{max} & \sum_{p\in\AA} \lambda(p) x(p)\\
		\text{subject to:}& \sum_{p\in\AA} n_p(\gamma) x(p) \le n(\gamma) & \forall \gamma\in\Gamma\\
		& x(p)\in\{0,1,\ldots,|\VV|\} & \forall p\in\AA
	\end{align*}

	We claim that the \KE instance is a \YES instance if and only if the optimal value of the above ILP is at least $t$.

	In one direction, suppose the \KE instance is a \YES instance. Let \CC be an optimal solution (a multi-set of paths and cycles) of the \KE instance. For signature sequence $p\in\AA$, we define $x(p) = \sum_{q\in\CC} \mathbbm{1}(p=\gamma(q))$ where $\mathbbm{1}(\cdot)$ is the indicator function. As $\CC$ is a subgraph of $\GG$, we have $\sum_{p\in\AA} n_p(\gamma)\sum_{q\in\CC} \mathbbm{1}(p=\gamma(q))\leq  n(\gamma)$, $\forall \gamma\in \Gamma$. Hence $x(p)_{p\in\AA}$ is a feasible solution. Since \CC covers at least $t$ non-altruistic vertices, it follows that $\sum_{p\in\AA} \lambda(p) x(p)\ge t$. Hence the optimal value of the above ILP is at least $t$.

	On the other direction, suppose there exists a solution $(x^*(p))_{p\in\AA}$ to the ILP such that $\sum_{p\in\AA} \lambda(p) x^*(p)\ge t$. We now describe an iterative approach to construct a solution \CC for the \KE instance from $(x^*(p))_{p\in\AA}$. We initialize \CC to the empty set and initialize $x^\pr(w)$ to 0 for all $w\in\AA$. Let $n_\CC(\gamma)$ denote the number of vertices of type $\gamma$ in \CC. Till there exists a $w\in\AA$ such that $x^\pr(w)<x^*(w)$, we add to $\CC$ a path/cycle $q$ of signature $w$ belonging to $\GG\setminus\CC$ and increase $x^\pr(w)$ by $1$.  Now we show that if there exists a $w\in\AA$ such that $x^\pr(w)<x^*(w)$, then there is always a path/cycle $q$ of signature $w$ in $\GG\setminus\CC$. Due to the way \AA is defined and the fact that $w\in\AA$, it suffices to show that $n_w(\gamma)\leq n(\gamma)-n_\CC(\gamma)$, $\forall\gamma\in\Gamma$. Since $(x^*(p))_{p\in\AA}$ is a feasible solution, $\sum_{p\in\AA} n_p(\gamma) x^*(p)\leq n(\gamma)$, $\forall \gamma \in \Gamma$. Therefore $n_w(\gamma)\leq\sum_{p\in\AA} n_p(\gamma) (x^*(p)-x^\pr(p))\leq n(\gamma)-\sum_{p\in\AA} n_p(\gamma) x^\pr(w)=n(\gamma)-n_\CC(\gamma)$, $\forall \gamma \in \Gamma$. Hence there is always a path/cycle $q$ of signature $w$. Now when the iterative procedure terminates, the number of non-altruistic vertices covered by \CC is $\sum_{p\in\AA} \lambda(p) x^*(p)$ which is at least $t$. Hence, the \KE instance is a \YES instance.	Now the result follows from \Cref{lenstra}.
\end{proof}

\subsection{Para-NP-Hardness Result}

In this section, we present an intractability result for \KE parameterized by $(\Delta+\el_p+\el_c)$.

\begin{definition}[X3C$^\prime$] Given a finite set $X = \{x_1, x_2,
\ldots , x_{3q}\}$ and a collection of 3-element subsets $C = \{(x_{i_h}, x_{j_h}, x_{k_h})| 1 \leq h\leq m\}$ of $X$ in which no element in $X$ appears in more than three subsets, compute if there exists a sub-collection $C’$ of $C$ such that every element in $X$ occurs in exactly one member of $C’$.
\end{definition}

It turns out that the proof of Abraham et al.~\cite[Theorem 1]{AbrahamBS07}, which shows \NP-completeness of the \KE problem, can be appropriately modified to get the \Cref{obs:dl}. We use X3C$^\prime$ problem which is known to be \NPC \cite{gonzalez1985clustering}.

\begin{observation}\label{obs:dl}
	The \KE problem, parameterized by $(\Delta+\el_p+\el_c)$, is \PNPH.
\end{observation}

\begin{proof}
	Consider an instance $(X,C)$ of X3C$^\prime$. For each subset in $C$, create a gadget as mentioned in the proof of Abraham et al.~\cite[Theorem 1]{AbrahamBS07}. Now set $t=9|C|+|X|$, $\el_p=0$ and $\el_c=3$. This completes the construction of the instance of the \KE problem. The correctness of the reduction follows from proof of Abraham et al.~\cite[Theorem 1]{AbrahamBS07}. The instance of \KE problem which gets created has $\el_p=0$, $\el_c=3$ and $\Delta\leq 6$. Hence, the \KE problem, parameterized by $\Delta+\el_p+\el_c$ is \PNPH.
\end{proof}

\subsection{Kernelization Results}

We now present our result on kernelization for \KE. Note that kernelization algorithms are often found quite useful in practice for drastically reducing the size of the input instance. We show that the \KE problem admits a polynomial kernel for the parameter $t+\Delta$ for every constant $\el_p$ and $\el_c$. Formally we show the following result.

\begin{theorem}\label{thm:kernel}
	For the \KE problem, there exists a vertex kernel of size $\OO(t\Delta^{\max\{\el_p,\el_c\}})$.
\end{theorem}

\begin{proof}
	Let $(\GG,\BB,\el_p,\el_c,t)$ be an arbitrary instance of the \KE problem. We apply the following reduction rule.

		{\it If there is a vertex $v\in\GG$ which participates neither in any cycle of length at most $\el_c$ nor in any path of length at most $\el_p$ starting from any vertex in \BB, then delete $v$.}

	The correctness of the reduction rule follows from the fact that any vertex that the above reduction rule deletes never participates in any solution (which is a collection of cycles and paths starting from any vertex in \BB of length at most $\el_c$ and $\el_p$, respectively).

 In order to find all the vertices which get removed due to the reduction rule, we do the following:
 \begin{enumerate}
 	\item Initialize an empty set $S$.
 	\item For each vertex  $v$, using Dijkstra’s algorithm, we find the shortest cycle and the shortest path starting from any vertex in \BB which contains $v$. If the length of such a cycle and path is greater than $\el_c$ and $\el_p$, respectively, then we add $v$ to the set $S$.
 	\item Delete all the vertices in $S$ and the graph induced on $\VV\setminus S$ is our reduced instance.
 \end{enumerate}
 These steps can be done in polynomial time as the shortest directed path and cycle in a graph can be found in polynomial time using Dijkstra’s algorithm. Note that the reduction rule is no longer applicable as no vertex from set $S$ was part of any cycle or path starting from any vertex in \BB of length at most $\el_c$ and $\el_p$, respectively, and their removal doesn't affect such cycles and paths.

When the reduction rule is no longer applicable, we show that the number of vertices in the graph is at most $O(t\Delta^ {\max\{\ell_p,\ell_c\}}).$ To see this, let us choose any maximal collection $C$ of cycles and paths (starting from an altruistic vertex) of length at most $\ell_c$ and $\ell_p$. Such a collection can be found by iteratively finding the shortest cycles and paths (starting from an altruistic vertex), adding it to the collection if they satisfy the constraints on the length and deleting them from the input graph. We stop this iterative process when can't find any such cycle and path. \textcolor{black}{This step can be done in polynomial time as the shortest directed path and cycle in a graph can be found in polynomial time using Dijkstra’s algorithm.} Let $W$ be the set of vertices in $C$. Note that the number of altruistic vertices in $W$ is at most $|W|/2$ as each path starting with an altruistic vertex has at least one non-altrusitic vertex and there are no isolated vertices in $W$. We assume $|W|<2t$, otherwise we output a trivial yes instance as the number of non-altruistic vertices in $W$ would be at least $t$. Since $C$ is maximal and the fact that the reduction rule is no longer applicable, every vertex in $\VV\setminus W$ should be within a distance of at most $\max\{\ell_p,\ell_c\}$ from at least one vertex of $C$ in the underlying undirected graph. It follows that the number of vertices in $\VV\setminus C$ is at most $\OO(t\Delta^ {\max\{\ell_p,\ell_c\}})$, which proves the result.
\end{proof}

Hence, \Cref{{thm:kernel}} shows that we can significantly reduce the size of the \KE problem instancs where both $\el_p$ and $\el_c$ are small, and the maximum degree is also small, to contain only $\OO(t)$ vertices. Note that there could be practical reasons to ensure both $\el_p$ and $\el_c$ to be small as pointed out by Dickerson et al.~\cite{dickerson2016position}. Is it possible to get rid of the exponential dependence on $\max\{\el_p,\el_c\}$ in the kernel size in \Cref{thm:kernel}? In particular, does the \KE problem admit a polynomial kernel parameterized by $t+\Delta+\el_p+\el_c$? We answer this question negatively in \Cref{thm:no-kernel}.


\begin{theorem}\label{thm:no-kernel}
The \KE problem does not admit any polynomial kernel with respect to the parameter $t+\Delta+\el_p+\el_c$ unless $\NP \subseteq \coNP\text{/\text{poly}}$.
\end{theorem}
\begin{proof}
	We show that the \HP problem on simple (no self loop) directed graphs cross-composes into the \KE problem. First, we take a polynomial equivalence relation $\mathcal{R}$ that places all malformed instances into one equivalence class, while all the well-formed instances are in the same equivalence class if the instances have the same number of vertices. Now, we present the cross composition algorithm. Given a set of malformed instances of \HP, we provide a trivial no instance of the \KE problem. Now let us assume that we are given $k$ well-formed instances $\GG_1, \ldots, \GG_k$ of \HP from an equivalence class of $\mathcal{R}$ as an input to the cross composition algorithm. Let the number of vertices in each of these $k$ instances be $n$. Let $k'$ be smallest power of 2 greater than or equal to $k$. Let \TT be a complete directed binary tree with $k'$ leaf nodes, root being the only altruistic vertex and the direction of edges in the tree is from the parent to the child. Let $h$ be the height (largest number of edges in a path from the root to a leaf node) of \TT. Note that $h=\OO(\log k)$ as $k'\leq 2k$. The cross-composition algorithm creates an instance \HH of the \KE problem by taking a disjoint union of $\GG_1, \ldots, \GG_k$, adding the tree \TT, and adding an outgoing edge from $i$-th leaf node to every vertex of the $\GG_i$. Note that the maximum degree $\Delta$ of \HH is $n+1$. Clearly, the \KE problem instance $(\HH,\BB=\{u\},\el_p=h+n,\el_c=0, t=h+n)$ is a \YES instance if and only if there exists an integer $i\in[k]$ such that the \HP instance $\GG_i$ is a \YES instance. Since \HP on simple directed graphs is an \NPC problem \cite{garey1979computers}, it follows that the \KE problem does not admit a polynomial kernel with respect to the parameter $t+\Delta+\el_p+\el_c$ unless \caveat~\cite{CyganFKLMPPS15}.
\end{proof}

Hence, if any of $\el_p$ or $\el_c$ is large, which can be the case sometimes in practice also, \Cref{thm:no-kernel} shows that we cannot hope to always significantly reduce the size of the input instance.

\subsection{Approximation Algorithm}

We now present our approximation algorithm for the \KE problem when only cycles of length at most $3$ are allowed; no path is allowed. Biro et al.~\cite{biro2009maximum} studied this problem with the name \MSTE and proved \APX-hardness. A trivial extension from the result on MAX CYCLE WEIGHT$\leq$k-WAY EXCHANGE in Biro et al.~\cite{biro2009maximum} leads to a 2+\eps approximation algorithm for \MSTE. Now, towards designing the approximation algorithm, we use the \TSP problem which is defined as follows.


\begin{definition}[\TSP]
Given a universe \UU, a family
$\mathcal{F} \subseteq 2^\UU$ of sets of size at most $3$, and an integer $k$, compute if there exists a subfamily $\WW\subseteq\mathcal{F}$ of pairwise
disjoint sets such that $|\WW|\ge k$. We denote any instance of \TSP by $(\UU,\FF,k)$.
\end{definition}

The following result is due to~\cite{cygan2013improved}.

\begin{lemma}~\cite{cygan2013improved}.\label{kset}
For every $\eps>0$, there is $(4/3+\eps)$-approximation algorithm for the \TSP problem for optimizing $k$.
\end{lemma}

\begin{theorem}\label{approx}
If there is a $\alpha$-approximation algorithm for \TSP problem, then there is a $\frac{4\alpha}{3}$-approximation algorithm for \MSTE.
\end{theorem}

\begin{proof}
Let $\GG$ denote the input graph associated with the problem instance of \MSTE. The high level idea of our algorithm is as follows. Using the graph \GG, we create two instances of the \TSP problem, solve them using the approximation algorithm for the \TSP problem, and output the best of the two (corresponding \MSTE) solutions.

{\bf Algorithm $1$:} Let us define $\FF_1=\{\{a,b,c\}\subseteq\VV[\GG]: a, b, c \text{ forms a triangle in }\GG\}$. Using the $\alpha$-approximation algorithm, we compute a solution $\SS_1$ for the \TSP instance $(\VV[\GG],\FF_1)$. Let $\CC_1$ be the set of cycles corresponding to the sets in $\SS_1$. Suppose the total number of vertices covered by the cycles in $\CC_1$ be $A_1$.

{\bf Algorithm $2$:} Let us define $\FF_2=\{\{a,b,c\}\subseteq\VV[\GG]: a, b, c \text{ forms a triangle in }\GG\}\cup\{\{a,b\}\subseteq\VV[\GG]: a,b \text{ forms a $2$ cycle in }\GG\}$. Using the $\alpha$-approximation algorithm, we compute a solution $\SS_2$ for the \TSP instance $(\VV[\GG],\FF_2)$. Let $\CC_2$ be the set of cycles corresponding to the sets in $\SS_2$. Suppose the total number of vertices covered by the cycles in $\CC_2$ be $A_2$.

If $A_1>A_2$, then our algorithm returns $\CC_1$; otherwise our algorithm returns $\CC_2$. Since $|\FF_1|, |\FF_2| = \OO(n^3)$, our algorithm runs in polynomial time.

We now prove approximation ratio of our algorithm. Towards that, let us fix an optimal solution $\CC_{\OPT}$ which covers \OPT number of vertices. Let $\OPT_i$ denote the number of vertices covered by cycles of length $i$ in $\CC_\OPT$ for $i=2,3$. Then, we have $\OPT=\OPT_2+\OPT_3$. The $\alpha$-approximation algorithm used in Algorithm 1 gives us the following.
\[\frac{A_1}{3} \ge \frac{\OPT_3}{3\alpha} \Rightarrow A_1 \ge \frac{\OPT_3}{\alpha}\]

Similarly, the $\alpha$-approximation algorithm used in Algorithm 2 gives us the following.
\begin{align*}
	\frac{A_2}{2} &\ge \frac{1}{\alpha}\left(\frac{\OPT_3}{3}+\frac{\OPT_2}{2}\right)\\
	\Rightarrow A_2 &\ge \frac{2}{\alpha}\left(\frac{\OPT_3}{3}+\frac{\OPT_2}{2}\right)\\
	&= \frac{2}{\alpha}\left(\frac{\OPT}{2}-\frac{\OPT_3}{6}\right) & [\OPT=\OPT_2+\OPT_3]\\
	&= \frac{1}{\alpha}\left(\OPT-\frac{\OPT_3}{3}\right)
\end{align*}

Hence, we can bound $\max\{A_1,A_2\}$ as follows which proves the claimed approximation guarantee of our algorithm.
\[ \max\{A_1,A_2\} \ge \frac{\OPT}{4\alpha/3} \]
\end{proof}

\Cref{approx} immediately gives us the following corollary.
\begin{corollary}
For every $\eps>0$, there is a $(\frac{16}{9}+\eps)$-approximation algorithm for the \KE problem if only cycles of length at most $3$ are allowed (and no paths are allowed).
\end{corollary}


\section{Conclusion and Open Questions}

In this paper, we have presented \FPT algorithms for the \KE problem with respect to some natural parameters, namely (i) solution size $t$ (the number of recipients receiving a kidney), (ii) treewidth+$\max\{\el_p,\el_c\}$ and (iii) the number of vertex types when $\el_p\leq\el_c$. We also depicted an MSO$_2$ formula for the \KE problem where the length of the formula is upper bounded by a function of $\max\{\el_p,\el_c\}$. For kernelization, we have exhibited a polynomial kernel w.r.t $\Delta+t$ when $\max\{\el_p,\el_c\}$ is $\OO(1)$. We have complemented this result by refuting existence of a polynomial kernel parameterized $\Delta+t+\max\{\el_p,\el_c\}$ unless $\NP \subseteq \coNP\text{/\text{poly}}$. We have finally presented an $(\frac{16}{9}+\epsilon)$-approximation algorithm for the \KE problem in the special case when cycles of length at most $3$ are allowed and no path is allowed.

Our work leaves many interesting open questions. One such question is whether the \KE problem is \FPT w.r.t treewidth only. Another such question is whether the \KE problem is \FPT w.r.t number of vertex types (without any assumption on $\el_c$ and $\el_p$). Another important question is the existence of a polynomial kernel for the \KE problem parameterized by the solution size alone when the maximum allowed length of paths and cycles are the same. Our hardness proof in \Cref{thm:no-kernel} breaks down if we want to allow paths and cycles of the same length. Finally, it is interesting to study whether the running time of our \FPT algorithms can be improved further or are they best possible assuming standard complexity-theoretic assumptions like ETH or SETH.

\bibliographystyle{alpha}
\bibliography{references}

\newcommand{\etalchar}[1]{$^{#1}$}
\begin{thebibliography}{GvdKW14}

\bibitem[ABS07]{AbrahamBS07}
David~J. Abraham, Avrim Blum, and Tuomas Sandholm.
\newblock Clearing algorithms for barter exchange markets: enabling nationwide
  kidney exchanges.
\newblock In Jeffrey~K. MacKie{-}Mason, David~C. Parkes, and Paul Resnick,
  editors, {\em Proc. 8th {ACM} Conference on Electronic Commerce (EC-2007)},
  pages 295--304. {ACM}, 2007.

\bibitem[ALS91]{arnborg1991easy}
Stefan Arnborg, Jens Lagergren, and Detlef Seese.
\newblock Easy problems for tree-decomposable graphs.
\newblock {\em Journal of Algorithms}, 12(2):308--340, 1991.

\bibitem[{Alv}04]{ke04}
{Alvin Roth, Tayfun S\"{o}nmez, Utku \"{U}nver}.
\newblock Kidney exchange.
\newblock {\em Quarterly Journal of Economics}, 119(2):457--488, 2004.

\bibitem[{Alv}05a]{roth05}
{Alvin Roth, Tayfun S\"{o}nmez, Utku \"{U}nver}.
\newblock A kidney exchange clearinghouse in new england.
\newblock {\em Am. Econ. Rev}, 95(2):376--380, 2005.

\bibitem[{Alv}05b]{ke05}
{Alvin Roth, Tayfun S\"{o}nmez, Utku \"{U}nver}.
\newblock Pairwise kidney exchange.
\newblock {\em J. Econ. Theory}, 125(2):151--188, 2005.

\bibitem[{Alv}07]{ke07}
{Alvin Roth, Tayfun S\"{o}nmez, Utku \"{U}nver}.
\newblock Efficient kidney exchange: Coincidence of wants in markets with
  compatibility-based preferences.
\newblock {\em Am. Econ. Rev}, pages 828--851, 2007.

\bibitem[AYZ94]{alon1994color}
Noga Alon, Raphy Yuster, and Uri Zwick.
\newblock Color-coding: a new method for finding simple paths, cycles and other
  small subgraphs within large graphs.
\newblock In {\em Proc. 26th annual ACM symposium on Theory of computing},
  pages 326--335, 1994.

\bibitem[AYZ95]{alon1995color}
Noga Alon, Raphael Yuster, and Uri Zwick.
\newblock Color-coding.
\newblock {\em Journal of the ACM (JACM)}, 42(4):844--856, 1995.

\bibitem[BHK{\etalchar{+}}22]{DBLP:journals/algorithmica/BelmonteHKKKKLO22}
R{\'{e}}my Belmonte, Tesshu Hanaka, Masaaki Kanzaki, Masashi Kiyomi, Yasuaki
  Kobayashi, Yusuke Kobayashi, Michael Lampis, Hirotaka Ono, and Yota Otachi.
\newblock Parameterized complexity of (a, l )-path packing.
\newblock {\em Algorithmica}, 84(4):871--895, 2022.

\bibitem[BMR09]{biro2009maximum}
P{\'e}ter Biro, David~F Manlove, and Romeo Rizzi.
\newblock Maximum weight cycle packing in directed graphs, with application to
  kidney exchange programs.
\newblock {\em Discrete Math Algorithms Appl}, 1(04):499--517, 2009.

\bibitem[BOD18]{BODLAENDER201842}
A faster parameterized algorithm for pseudoforest deletion.
\newblock {\em Discrete Applied Mathematics}, 236:42--56, 2018.

\bibitem[CFK{\etalchar{+}}15]{CyganFKLMPPS15}
Marek Cygan, Fedor~V. Fomin, Lukasz Kowalik, Daniel Lokshtanov, D{\'{a}}niel
  Marx, Marcin Pilipczuk, Michal Pilipczuk, and Saket Saurabh.
\newblock {\em Parameterized Algorithms}.
\newblock Springer, 2015.

\bibitem[CKVR13]{constantino2013new}
Miguel Constantino, Xenia Klimentova, Ana Viana, and Abdur Rais.
\newblock New insights on integer-programming models for the kidney exchange
  problem.
\newblock {\em Eur. J. Oper. Res.}, 231(1):57--68, 2013.

\bibitem[Cyg13]{cygan2013improved}
Marek Cygan.
\newblock Improved approximation for 3-dimensional matching via bounded
  pathwidth local search.
\newblock In {\em 2013 IEEE 54th Annual Symposium on Foundations of Computer
  Science}, pages 509--518. IEEE, 2013.

\bibitem[DMP{\etalchar{+}}16]{dickerson2016position}
John~P Dickerson, David~F Manlove, Benjamin Plaut, Tuomas Sandholm, and James
  Trimble.
\newblock Position-indexed formulations for kidney exchange.
\newblock In {\em Proc. 2016 ACM Conference on Economics and Computation},
  pages 25--42, 2016.

\bibitem[GJ79]{garey1979computers}
Michael~R Garey and David~S Johnson.
\newblock {\em Computers and intractability}, volume 174.
\newblock freeman San Francisco, 1979.

\bibitem[Gon85]{gonzalez1985clustering}
Teofilo~F. Gonzalez.
\newblock Clustering to minimize the maximum intercluster distance.
\newblock {\em Theor. Comput. Sci.}, 38:293--306, 1985.

\bibitem[GvdKW14]{glorie2014kidney}
Kristiaan~M Glorie, J~Joris van~de Klundert, and Albert~PM Wagelmans.
\newblock Kidney exchange with long chains: An efficient pricing algorithm for
  clearing barter exchanges with branch-and-price.
\newblock {\em Manuf. Serv. Oper. Manag.}, 16(4):498--512, 2014.

\bibitem[JTWZ17]{jia2017efficient}
Zhipeng Jia, Pingzhong Tang, Ruosong Wang, and Hanrui Zhang.
\newblock Efficient near-optimal algorithms for barter exchange.
\newblock In {\em Proc. 16th Conference on Autonomous Agents and MultiAgent
  Systems}, pages 362--370, 2017.

\bibitem[KAV14]{klimentova2014new}
Xenia Klimentova, Filipe Alvelos, and Ana Viana.
\newblock A new branch-and-price approach for the kidney exchange problem.
\newblock In {\em Proc. International Conference on Computational Science and
  Its Applications}, pages 237--252. Springer, 2014.

\bibitem[KNS{\etalchar{+}}07]{krivelevich2007approximation}
Michael Krivelevich, Zeev Nutov, Mohammad~R Salavatipour, Jacques Verstraete,
  and Raphael Yuster.
\newblock Approximation algorithms and hardness results for cycle packing
  problems.
\newblock {\em ACM T Algorithms}, 3(4):48--es, 2007.

\bibitem[LJ83]{lenstra1983integer}
Hendrik~W Lenstra~Jr.
\newblock Integer programming with a fixed number of variables.
\newblock {\em Mathematics of operations research}, 8(4):538--548, 1983.

\bibitem[LLHT14]{li2014egalitarian}
Jian Li, Yicheng Liu, Lingxiao Huang, and Pingzhong Tang.
\newblock Egalitarian pairwise kidney exchange: fast algorithms via linear
  programming and parametric flow.
\newblock In {\em AAMAS}, pages 445--452, 2014.

\bibitem[LWFF19]{lin2019randomized}
Mugang Lin, Jianxin Wang, Qilong Feng, and Bin Fu.
\newblock Randomized parameterized algorithms for the kidney exchange problem.
\newblock {\em Algorithms}, 12(2):50, 2019.

\bibitem[MD22]{DBLP:conf/ijcai/MaitiD22}
Arnab Maiti and Palash Dey.
\newblock Parameterized algorithms for kidney exchange.
\newblock In Luc~De Raedt, editor, {\em Proceedings of the Thirty-First
  International Joint Conference on Artificial Intelligence, {IJCAI} 2022,
  Vienna, Austria, 23-29 July 2022}, pages 405--411. ijcai.org, 2022.

\bibitem[MH17]{mak2017kidney}
Vicky Mak-Hau.
\newblock On the kidney exchange problem: cardinality constrained cycle and
  chain problems on directed graphs: a survey of integer programming
  approaches.
\newblock {\em J Comb Optim}, 33(1):35--59, 2017.

\bibitem[MO15]{manlove2015paired}
David~F Manlove and Gregg O’malley.
\newblock Paired and altruistic kidney donation in the uk: Algorithms and
  experimentation.
\newblock {\em J. Exp. Algorithmics}, 19:1--21, 2015.

\bibitem[odt]{odts}
{\em Organ Donation and Transplantation Statistics}.
\newblock
  \url{https://www.kidney.org/news/newsroom/factsheets/Organ-Donation-and-Transplantation-Stats}.

\bibitem[Rap86]{rapaport1986case}
Felix~T Rapaport.
\newblock The case for a living emotionally related international kidney donor
  exchange registry.
\newblock In {\em Transplantation proceedings}, volume~18, page~5, 1986.

\bibitem[RAR15]{anderson}
David~Gamarnik Ross~Anderson, Itai~Ashlagi and Alvin~E Roth.
\newblock Finding long chains in kidney exchange using the traveling salesman
  problem.
\newblock {\em Proc. National Academy of Sciences}, 112(3):663--668, 2015.

\bibitem[SGW{\etalchar{+}}05]{segev2005kidney}
Dorry~L Segev, Sommer~E Gentry, Daniel~S Warren, Brigitte Reeb, and Robert~A
  Montgomery.
\newblock Kidney paired donation and optimizing the use of live donor organs.
\newblock {\em Jama}, 293(15):1883--1890, 2005.

\bibitem[S{\"U}14]{sonmez2014altruistically}
Tayfun S{\"o}nmez and M~Utku {\"U}nver.
\newblock Altruistically unbalanced kidney exchange.
\newblock {\em J. Econ. Theory}, 152:105--129, 2014.

\bibitem[uno]{unos}
{\em United Network for Organ Sharing (UNOS)}.
\newblock \url{https://unos.org/}.

\bibitem[usr]{usrds}
{\em United States Renal Data System (USRDS)}.
\newblock \url{https://www.usrds.org/}.

\bibitem[XW18]{xiao2018exact}
Mingyu Xiao and Xuanbei Wang.
\newblock Exact algorithms and complexity of kidney exchange.
\newblock In {\em IJCAI}, pages 555--561, 2018.

\bibitem[Zeh16]{zehavi2016randomized}
Meirav Zehavi.
\newblock A randomized algorithm for long directed cycle.
\newblock {\em Information Processing Letters}, 116(6):419--422, 2016.

\end{thebibliography}

\end{document}